\documentclass[kpfonts]{patmorin}
\usepackage{pat}
\usepackage{paralist}
\usepackage{dsfont}  
\usepackage[utf8]{inputenc}
\usepackage[T1]{fontenc}

\usepackage{graphicx}
\usepackage[noend]{algorithmic}

\usepackage{xcolor}
\definecolor{light-gray}{gray}{0.95}

\usepackage[normalem]{ulem}
\usepackage{cancel}
\usepackage{enumitem}

\let\le\leqslant
\let\ge\geqslant

\title{\MakeUppercase{A Fast Algorithm for the Product Structure of Planar Graphs}}
\author{
  Pat Morin%
    \thanks{School of Computer Science, Carleton University, Canada. This research was partially supported by NSERC.}
}
\date{}

\begin{document}
\maketitle

\begin{abstract}
  Dujmović \etal\  (FOCS2019) recently proved that every planar graph $G$ is a subgraph of $H\boxtimes P$, where $\boxtimes$ denotes the strong graph product, $H$ is a graph of treewidth 8 and $P$ is a path.  This result has found numerous applications to linear graph layouts, graph colouring, and graph labelling.  The proof given by Dujmović \etal\  is based on a similar decomposition of Pilipczuk and Siebertz (SODA2019) which is constructive and leads to an $O(n^2)$ time algorithm for finding $H$ and the mapping from $V(G)$ onto $V(H\boxtimes P)$.  In this note, we show that this algorithm can be made to run in $O(n\log n)$ time.
\end{abstract}

\section{Introduction}

The \emph{strong product} $G_1\boxtimes G_2$ of two graphs $G_1$ and $G_2$ is the graph whose vertex set is the Cartesian product $V(G_1)\times V(G_2)$ in which the vertices $(v,x)$ and $(w,y)$ are adjacent if and only if
\begin{compactitem}
  \item $v=w$ and $xy\in E(G_2)$;
  \item $vw\in E(G_1)$ and $x=y$; or
  \item $vw\in E(G_1)$ and $xy\in E(G_2)$.
\end{compactitem}
Dujmović \etal\  \cite{dujmovic.joret.ea:planar} recently proved the following \emph{product structure theorem} for planar graphs:

\begin{thm}[Dujmović \etal\ 2019]\thmlabel{product-structure}
  For any $n$-vertex planar graph $G$, there exists a graph $H$ of treewidth at most 8 and a path $P$ such that $G$ is a subgraph of $G^+:=H\boxtimes P$.
\end{thm}

Though still very new, \thmref{product-structure} has been used to solve a number of longstanding open problems on planar graphs:
\begin{itemize}
  \item \thmref{product-structure} has been used to show that the queue-number of every planar graph is upper bounded by a constant.  This solves an open problem of Heath, Leighton, and Rosenberg posed in 1992 \cite{heath.leighton.ea:comparing}.
  \item \thmref{product-structure} has been used to show that the nonrepetitive chromatic number of every planar graph is upper bounded by a constant. This solves an open problem of Alon \etal\  \cite{alon.grytczuk.ea:nonrepetitive} posed in 2002.
  \item \thmref{product-structure} has been used to produce (asymptotically) optimal labelling schemes for planar graphs \cite{dujmovic.esperet.ea:adjacency}.  This (asympotically) resolves a problem of Kannan, Naor, and Rudich posed in 1988 \cite{kannan.naor.ea:implicit-stoc,kannan.naor.ea:implicit}.
  \item \thmref{product-structure} has been used to make significant improvements on the best-known bounds for $p$-centered colourings of planar graphs \cite{debski.felsner.ea:improved}.  This gives the strongest result thus far on a question motivated by the work of Nešetřil and Ossona de Mendez from 2006 \cite{nesetril.ossona:tree,nesetril.ossona:grad} and posed explicity by Dvořák in 2016 \cite{dvorak:question}.
\end{itemize}

The proof of \thmref{product-structure} given by Dujmović \etal\  is based on a similar decomposition of Pilipczuk and Siebertz \cite{pilipczuk.siebertz:polynomial-soda} which is constructive and leads to an $O(n^2)$ time algorithm for finding $H$ and the mapping from $V(G)$ onto $V(H\boxtimes P)$ \cite[Section~10]{dujmovic.joret.ea:planar}. Given the number of applications of \thmref{product-structure} (and that more are likely to be found), it is natural to ask if this running-time can be improved.  In this paper, we provide a faster algorithmic version of \thmref{product-structure}:

\begin{thm}\thmlabel{main}
  For any $n$-vertex planar graph $G$, there exists a graph $H$ of treewidth at most 8 and a path $P$ such that $G$ is a subgraph of $G^+:=H\boxtimes P$.

  Furthermore, there exists an algorithm that, given $G$ as input, runs in $O(n\log n)$ time and produces the graph $H$, the path $P$, and an injective function $\varphi:V(G)\to V(G^+)$ such that, for each edge $vw\in E(G)$,  $\varphi(v)\varphi(w)\in E(G^+)$.
\end{thm}

The remainder of this paper is organized as follows. \Secref{original} reviews the proof of \thmref{product-structure} and the resulting $O(n^2)$ time algorithm.  \Secref{faster} describes the $O(n\log n)$ time algorithm.  \Secref{discussion} discusses some of the implications and generalizations of this work.

\section{The Original Proof/Algorithm}
\seclabel{original}

Throughout this paper we use standard graph theory terminology as used in the textbook by Diestel \cite{diestel:graph}.  Every graph $G$ that we consider is finite, simple, and undirected, and has vertex set denoted by $V(G)$ and edge set denoted by $E(G)$.

Let $T$ be a tree rooted at some node $r$ and, for each node $v$ of $T$, let $P_T(v)$ denote the path in $T$ from $v$ to $r$.  The \emph{$T$-depth} of a node $v$ in $T$ is the length of $P_T(v)$.\footnote{The \emph{length} of a path is equal to the number of edges in the path, which is one less than the number of vertices in the path.}  A path $P$ in $T$ is a \emph{vertical path} if no two nodes of $P$ have the same $T$-depth.  Every node $w$ in $P_T(v)$ is a \emph{$T$-ancestor} of $v$ and $v$ is a \emph{$T$-descendant} of every node $w$ in $P_T(v)$.  Note that $v$ is both a $T$-ancestor and $T$-descendant of itself.

For a graph $G$ and a partition $\mathcal{P}$ of $V(G)$, the \emph{quotient graph} $G/\mathcal{P}$ is the graph whose vertices $V(G/\mathcal{P})$ are the sets in $\mathcal{P}$ and in which an edge $XY\in E(G/\mathcal{P})$ if and only if there exists $x\in X$ and $y\in Y$ with $xy\in E(G)$.  Dujmović \etal\  \cite{dujmovic.joret.ea:planar} prove \thmref{product-structure} by first adding edges to a planar graph $G_0$ to complete it to a triangulation $G$, computing a breadth-first search tree $T$ of $G$ and then applying the following result to $G$ and $T$:

\begin{thm}\thmlabel{triangulation-partition}
  For any $n$-vertex triangulation $G$ and any spanning tree $T$ of $G$, there exists a partition $\mathcal{P}$ of $V(G)$ such that each $P\in\mathcal{P}$ induces a vertical path in $T$ and the quotient graph $H:=G/\mathcal{P}$ has treewidth at most $8$.
\end{thm}

Deriving \thmref{product-structure} from \thmref{triangulation-partition} is just a matter of checking definitions.  The graph $H$ in \thmref{product-structure} is the same graph $H$ in \thmref{triangulation-partition}. The path $P$ in \thmref{product-structure} is simply the path $0,1,2,\ldots,h$ where $h$ is the maximum depth of any node in $T$.  Each vertex $v\in V(G)$ maps to the node $\varphi(v):=(X,y)$ where $X$ is the set in $\mathcal{P}$ that contains $v$ and $y$ is the depth of $v$ in $T$.  It is straightforward to check (using the definition of $\boxtimes$ and the fact that $T$ is a breadth-first search tree) that for any  edge $vw\in E(G)$, $\varphi(v)\varphi(w)\in E(H\boxtimes P)$.

Therefore, we will focus on giving a fast algorithm for \thmref{triangulation-partition}, from which we immediately obtain \thmref{main}.  We begin by describing the proof of Dujmović \etal\  \cite{dujmovic.joret.ea:planar}, which is inductive, and leads naturally to a recursive algorithm.  Refer to \figref{tripod}. The algorithm is initialized with a breadth-first-search tree $T$ of the triangulation $G$.  Each recursive invocation of the algorithm is given as input:
\begin{compactenum}
  \item A cycle $F$ in $G$.

  The subgraph of $G$ that includes the edges and vertices of $F$ and the edges and vertices of $G$ in the interior of $F$ is a near-triangulation, $N$.  The following are preconditions on the cycle $F$:
  \begin{compactenum}[(P1)]
    \item The root $r$ of $T$ is not in the interior of $F$, i.e., $r\not\in\ V(N)\setminus V(F)$.
    \item For every vertex $v\in V(N)\setminus V(F)$, and every $T$-descendant $w$ of $v$, $w\in V(N)\setminus V(F)$.
    \item Prior to this recursive invocation, every vertex of $F$ is already included in some part of the partition $\mathcal{P}$ and no vertex in $V(N)\setminus V(F)$ is included in any part of $\mathcal{P}$.
  \end{compactenum}
  \item Three edges $e_1$, $e_2$, and $e_3$ of $F$ that we will call \emph{portals}.

  Removing $e_1$, $e_2$ and $e_3$ from $F$ splits $F$ into three non-empty paths $P_1$, $P_2$, and $P_3$ where, for each $i\in\{1,2,3\}$, neither endpoint of $e_i$ is included in $P_i$.  The portals satisfy the following precondition:
  \begin{compactenum}[(P1)]\setcounter{enumii}{3}
    \item For each $i\in\{1,\ldots,3\}$, $V(P_i)$ is contained in the union of at most two elements of $\mathcal{P}$.
  \end{compactenum}

\end{compactenum}

\begin{figure}
  \begin{center}
    \includegraphics{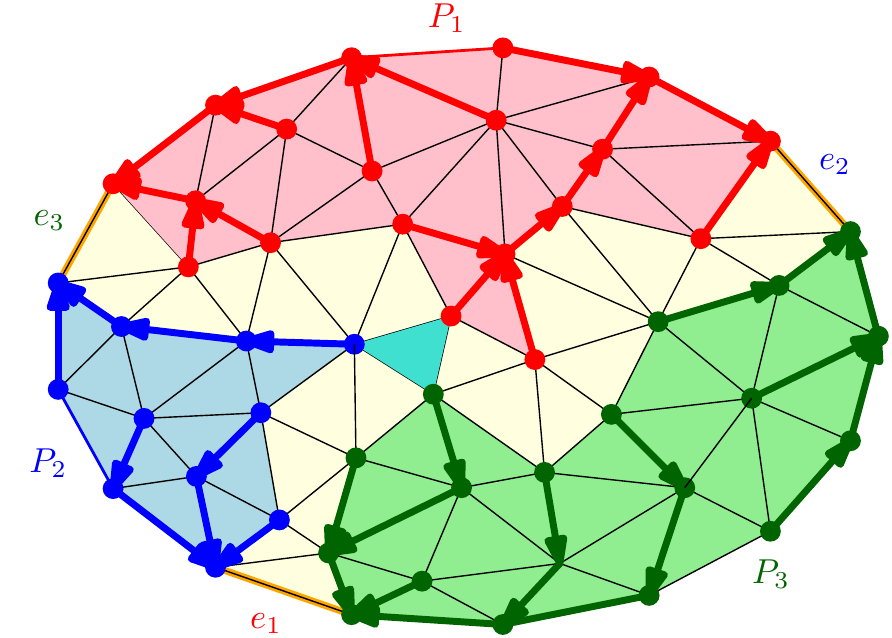} \\[1ex]
    \includegraphics{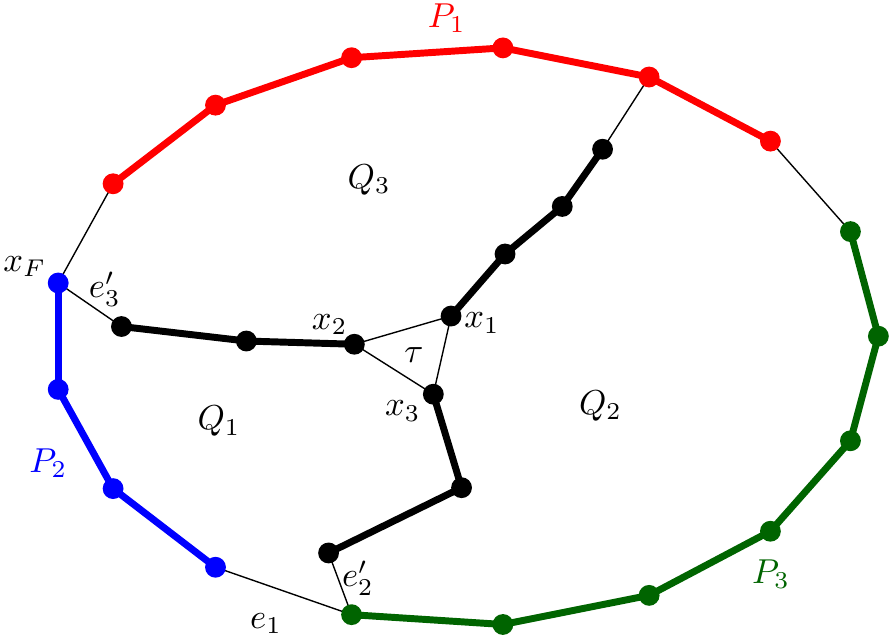} \\[1ex]
    \includegraphics{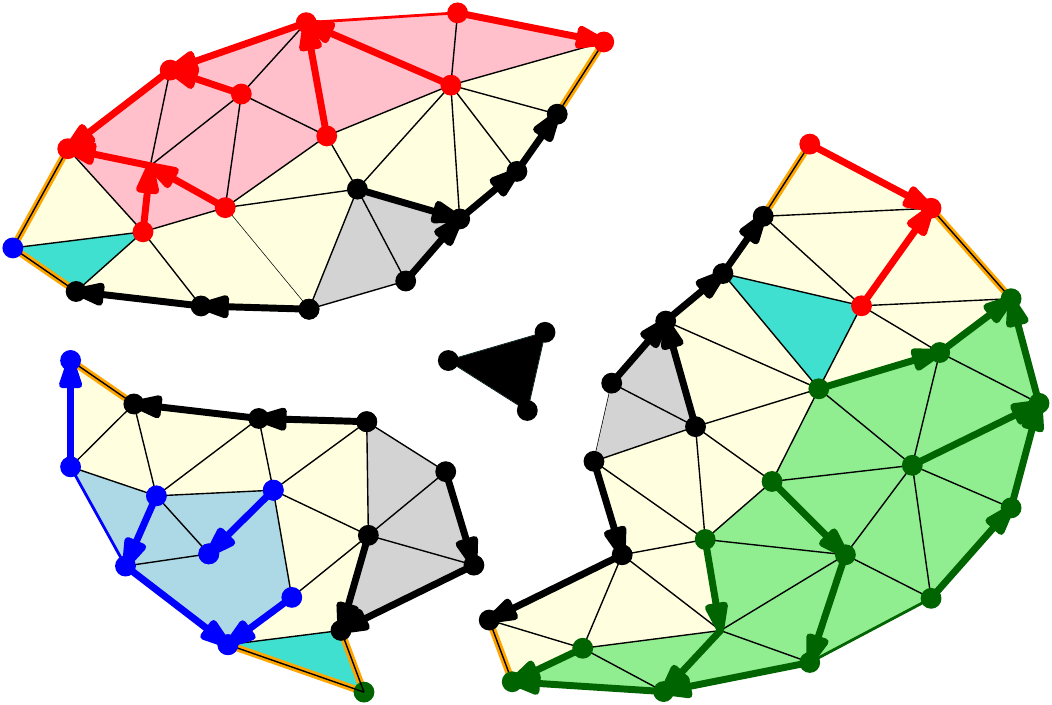}
  \end{center}
  \caption{A single recursive step from Dujmović \etal\  \cite{dujmovic.joret.ea:planar}.}
  \figlabel{tripod}
\end{figure}

By the time the recursive invocation terminates, each vertex of $N-V(F)$ is included in some part of the partition $\mathcal{P}$. Let $f$ denote the number of inner triangular faces of $N$.  The base case occurs when $f=1$ so $N$ consists of a single triangle ($F$).  In this case (P3) implies that each vertex of $N$ is already included in $\mathcal{P}$ and there is nothing to do so the algorithm returns immediately.

If $f>1$, the paths $P_1$, $P_2$, and $P_3$, along with the breadth-first search tree $T$ are used to partition the vertices of $N$ into three colour classes, as follows.  Each vertex $v\in P_i$ has colour $c(v)=i$.  For each vertex $v\in V(N)\setminus V(F)$, (P1) implies that $P_T(v)$ contains some first vertex $v_F$ of $F$.  The vertex $v$ is assigned the colour $c(v)=c(v_F)$.

By Sperner's Lemma, $N$ contains a triangular face $\tau=x_1x_2x_3$ that is \emph{trichromatic}, i.e., $c(x_i)=i$ for each $i\in\{1,2,3\}$. (Note that $0$, $1$, $2$, or $3$ vertices of $\tau$ may be in $V(F)$.)  The edges of $F$, $\tau$, and the paths in $T$ from each $x_i$ to the first vertex of $P_i$ define a graph $M$ with at most 4 interior faces, one of which is $\tau$.  Each of the other (at most three) interior faces does not contain $x_i$ for some $i\in\{1,2,3\}$. For each $i\in\{1,2,3\}$, we let $Q_i$ denote the interior face of $M$ that does not contain $x_i$. Observe that, for each $i\in\{1,2,3\}$, $Q_i$ contains no vertex of $P_i$.

For each $i\in\{1,2,3\}$, let $Z_i$ be the path, in $T$, from $x_i$ up to, but not including, the first vertex in $P_i$.  Note that $Z_i$ may be empty, which occurs when $x_i$ is a vertex of $P_i$.  Let $Y:=V(Z_1)\cup V(Z_2)\cup V(Z_3)$.  The algorithm adds $V(Z_1)$, $V(Z_2)$ and $V(Z_3)$ to the partition $\mathcal{P}$ and then recurses on each of $Q_1$, $Q_2$, and $Q_3$.

We now argue that $Q_1$ satisifies preconditions (P1)--(P3). The face $Q_1$ is a cycle in $G$ that is contained in the cycle $F$, so $Q_1$ satisifies precondition (P1).  The vertices of $Q_1$ are contained in $V(P_2)\cup V(P_3)\cup Y$.  Therefore every vertex of $Q_1$ is contained in some part of $\mathcal{P}$, so $Q_1$ satisfies precondition (P3).

Let $v$ be any vertex in the interior of $Q_1$ and let $w$ be any $T$-descendant of $v$.  In order to show that $Q_1$ satisfies precondition (P2) we must show that $w$ is in the interior of $Q_1$.  If $w$ is not in the interior of $Q_1$ then the path in $T$ from $w$ to $v$ contains a vertex $v'\in V(Z_2)\cup V(Z_3)\cup V(F)$.  Since $v'$ is a $T$-descendant of $v$, precondition (P2) implies that $v'\not\in V(F)$.  If $v'\in V(Z_i)$ for some $i\in\{2,3\}$, then the path from $w$ to $v$ in $T$ contains the subpath of $Z_i$ beginning at $v'$ and continuing to the last vertex $v''$ of $Z_i$.  The path from $w$ to $v$ in $T$ also contains the $T$-parent $v'''$ of $v''$.  This again contradicts (P2) because, by definition, $v'''\in V(F)$ and is a $T$-descendant of $v$.  Therefore every $T$-descendant $w$ of $v$ is contained in the interior of $Q_1$, so $Q_1$ satisfies precondition (P2).

Next we describe the three portals used when recursing on $Q_1$.
The cycle $Q_1$ contains at least one vertex each from $V(P_2)$ and $V(P_3)$ and therefore also contains the portal $e_1$, which is also used as one of the three portals in the recursive invocation.  If $V(Z_2)\cup V(Z_3)$ is non-empty, then $Q_1$ contains two edges $e_2'$ and $e_3'$ where $e_2'$ has an endpoint in $V(P_3)$ and an endpoint in $V(Z_2)\cup V(Z_3)$ and where $e_3'$ has an endpoint in $V(P_2)$ and an endpoint in $V(Z_2)\cup V(Z_3)$.  In this case, the edges $e_1$, $e_2'$, and $e_3'$ are used as the three portals in the recursive invocation on $Q_1$.  Note that $e_1$, $e_2'$ and $e_3'$ satisfy precondition (P4) since the vertices of $P_1'$---the path from $e_2'$ to $e_3'$ on $Q_1$ that does not contain $e_1$---are contained in the union of $V(Z_2)$ and $V(Z_3)$, which are included in $\mathcal{P}$.

If $(V(Z_2)\cup V(Z_3))$ is empty---because $x_2\in V(P_2)$ and $x_3\in V(P_3)$---then we artifically create two portals $e_2'$ and $e_3'$ for the recursive invocation by taking any two edges of $Q_1$ other than $e_1$. Clearly, this choice of $e_2'$ and $e_3'$ also satisfies precondition (P4).

The recursive invocations on $Q_2$ and $Q_3$ are done similarly, but rotating the values $1,2,3$.  After these three recursive invocations, every vertex in $N-V(F)$ is included in some part of $P$, so the recursive invocation is complete.

\subsection{Running-Time Analysis}

Recall that $f$ denotes the number of inner faces in the near-triangulation $N$.  By having each vertex of $G$ store a pointer to its parent in $T$ and storing $G$ using a representation that simultaneously represents $G$ and its dual graph $G^*$, the colouring of the vertices of $N$ can be done in $O(f)$ time and then the inner triangular faces of $N$ can be traversed in $O(f)$ time to find the trichromatic triangle $\tau$. The rest of the work (adding $Z_1$, $Z_2$, and $Z_3$ to $P$ and preparing the recursive invocations on $Q_1$, $Q_2$, and $Q_3$) is also easily implemented in $O(f)$ time, so the running time of the algorithm is given by the recurrence
\[  T(f) \le \begin{cases}
           a & \text{for $f\le 1$} \\
           a\cdot f + T(f_1)+T(f_2)+T(f_3) & \text{for $f\ge 2$}
         \end{cases}
 \]
where $a$ is a sufficiently large constant and, for each $i\in\{1,\ldots,3\}$, $f_i$ is the number of faces of $G$ contained in the interior of $Q_i$.
Note that $f_1+f_2+f_3=f-1$ (since $\tau$ is not contained in $Q_1$, $Q_2$, or $Q_3$).  An easy inductive proof shows that $T(f) \le a\cdot f\cdot (f+1)/2 = O(f^2)$.

The recursive procedure described above is used to prove \thmref{triangulation-partition} as follows.  Given an $n$-vertex triangulation $G$ and a spanning tree $T$ of $G$:
\begin{enumerate}
  \item Define one of the faces incident to the root $r$ of $T$ to be the outer face of $G$ and let $r$, $x$, and $y$ denote the three vertices on the outer face of $G$.
  \item Place $\{r\}$, $\{x\}$, and $\{y\}$ in the partition $\mathcal{P}$ and run the recursive procedure described above on the cycle $F:=rxy$ with the portals $e_1=rx$, $e_2=xy$ and $e_3=yr$.
\end{enumerate}
The first step of this procedure runs in constant time.  The second step requires $\Theta(f^2)=\Theta(n^2)$ time in the worst case.

\subsection{Treewidth Analysis}
\seclabel{treewidth-analysis}

Dujmović \etal\ \cite{dujmovic.joret.ea:planar} show that the contraction $H:=N/\mathcal{P}$ has treewidth at most 8.  Although it is not necessary to repeat their argument here, it is worth doing so because it illustrates that a width-$\le\!8$ tree decomposition of $H$ can be easily computed during the construction of the partition $\mathcal{P}$.  (See Diestel \cite[Chapter 12]{diestel:graph} for definitions of treewidth and tree decompositions.)

The argument mirrors the inductive structure of the algorithm used to create $\mathcal{P}$.  Specifically, the recursive invocation on $F$ takes as input the set $X$ (guaranteed by (P4)) of at most 6 parts of $\mathcal{P}$ that cover $V(F)$ and produces a tree decomposition of $N/\mathcal{P}$ that has some bag $B$ containing $X$ and in which every bag has size at most $9$.

For each $i\in\{1,2,3\}$, the at most 4 parts in $X$ that cover $V(F-P_i))$ (again, guaranteed by (P4)) and the two parts $V(Z_j)$, $j\in\{1,2,3\}\setminus\{i\}$ are used as the input $X_i$ to the recursive call on the cycle $Q_i$ that bounds the near-triangulation $N_i$. This produces a tree decomposition of $N_i/\mathcal{P}$ in which some bag $B_i$ contains $X_i$ and every bag has size at most $9$.

The three tree decompositions of $N_1$, $N_2$, and $N_3$ are then joined by introducing a bag $B:=X\cup\{V(Z_1),V(Z_2),V(Z_3)\}$ and making $B$ adjacent to $B_i$ for each $i\in\{1,2,3\}$.  It is straightforward to verify that this does indeed give a tree decomposition of $N/\mathcal{P}$.  The bag $B$ has size $|X|+3\le 9$ and (inductively) every other bag has size at most 9, thus proving that the treewidth of $N/\mathcal{P}$ is at most 8.

We note that this tree decomposition of $H$ can be computed at the same time as the partition $\mathcal{P}$ without contributing more than $O(|V(H)|)\subseteq O(n)$ to the running time of the algorithm.

\section{A Faster Algorithm}
\seclabel{faster}

To obtain a faster algorithm we will create an algorithm (part of) whose running time satisifies the recurrence:
\[  T(f) \le \begin{cases}
         a & \text{for $f\le 1$} \\
         a\cdot(1+\min\{f_1,f_2,f_3\}) + T(f_1)+T(f_2)+T(f_3) & \text{for $f\ge 2$}
       \end{cases}
\]
It is straightforward to show, by induction, that $T(f)\le (a/3)f\log_3(f)=O(a\cdot f\log f)$.  The value of $a$ here depends on the running time of an operation on a certain data structure described next.

Our algorithm makes use of a data structure that preprocesses a $(n+1)$-vertex tree $T$ with root $r$ so that it can maintain a set $S\subseteq V(T)$ that, initially, contains only $r$, and supports the following operations that each take a node $w\in V(T)$ as an argument:

\begin{itemize}
    \item $\textsc{Mark}(w)$: Add $w$ to the set $S$.  A precondition of this operation is that the parent, $v$, of $w$ is already in $S$ but $w$ is not yet in $S$.

    \item $\textsc{NearestMarkedAncestor}(w)$:  Return the first node $v\in S$ that is on the path from $w$ to the root of $T$.
\end{itemize}

In \appref{data_structures} we show how to use standard techniques to obtain the following result:\footnote{A data structure supporting these two operations in $O(1)$ amortized time per operation can be obtained from the work of Gabow and Tarjan \cite{gabow.tarjan:linear}, but the data structure described in \appref{data_structures} is considerably simpler to implement and is fast enough for our purposes.}

\begin{lem}\lemlabel{nearest_marked_ancestor}
    There exists a data structure that preprocesses any $n$-node rooted tree $T$ and supports the $\textsc{Mark}(w)$ and $\textsc{NearestMarkedAncestor}(w)$ operations. Each $\textsc{NearestMarkedAncestor}(w)$ operation takes $O(1)$ time and any sequence of $\textsc{Mark}(w)$ operations takes a total of $O(n\log n)$ time.
\end{lem}

In the remainder of this section, we will show how \lemref{nearest_marked_ancestor} can be used to achieve the desired running time.  It is worth noting that the algorithm we now describe has the same recursion tree and produces exactly the same partition $\mathcal{P}$ as the original algorithm.  Therefore $\mathcal{P}$ has all the properties described by Dujmović \etal.\ \cite{dujmovic.joret.ea:planar}. In particular, the quotient graph $H:=G/\mathcal{P}$ has treewidth at most $8$ and a tree decomposition of $H$ of widtch at most $8$ can be computed while computing $\mathcal{P}$.

As before, each recursive step takes as input the cycle $F$ and the three portals $e_1$, $e_2$, and $e_3$.  Additionally, the algorithm requires that the vertices of $P_1$, $P_2$ and $P_3$ are coloured with three different colours.  More precisely, there are three distinct integers $c_1$, $c_2$ and $c_3$ such that $c(v)=c_i$ for each $v\in V(P_i)$ and each $i\in\{1,2,3\}$.  The nearest marked ancestor data structure maintaining $S\subseteq V(T)$ is set up so that $V(F)\subseteq S$ and $(V(N)\setminus V(F))\cap S=\emptyset$.  That is, $S$ contains all vertices on the outer face of $N$, but none of the inner vertices.

The algorithm searches for the trichromatic triangle $\tau$ beginning from the portals.  Refer to \figref{fast-search}. Step~0 of the search begins with $e_{i,0}=e_i$ and $t_{i,0}$ as the unique triangular inner face of $N$ with $e_i$ on its boundary, for each $i\in\{1,2,3\}$. In Step~$j$ of the search, the algorithm has three triangles $t_{i,j}$ and three edges $e_{i,j}$ where $e_{i,j}$ is an edge of $t_{i,j}$ for each $i\in\{1,2,3\}$.  Using the data structure for $T$, the algorithm checks, for each $i\in\{1,2,3\}$, the colours of $t_{i,j}$'s three vertices by calling $\textsc{NearestMarkedAncestor}(w)$ for each of $t_{i,j}$'s three vertices.  This returns a vertex $v\in V(F)$ whose colour gives the colour of $w$.

If $t_{i,j}$ is trichromatic for at least one $i\in\{1,2,3\}$, then the algorithm has found the necessary trichromatic triangle $\tau$ and this step is complete. Otherwise, for each $i\in\{1,2,3\}$, the triangle $t_{i,j}$ contains another bichromatic edge $e_{i,j+1}\neq e_{i,j}$ and this edge bounds another triangular face $t_{i,j+1}\neq t_{i,j}$ of $N$.  The algorithm then continues to Step~$(j+1)$ of the search using the triangles $t_{i,j+1}$ and edges $e_{i,j+1}$ for each $i\in\{1,2,3\}$.  The fact that this algorithm terminates (and would terminate even if the search were limited to any one of the portals) follows from a classic proof of Sperner's Lemma in two dimensions.

 \begin{figure}
   \begin{center}
     \includegraphics{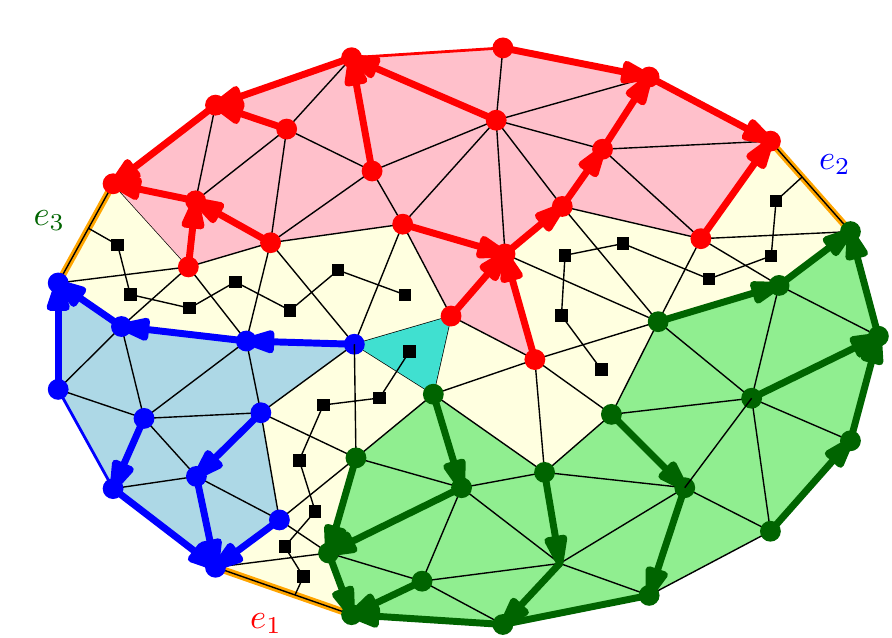} \\[1ex]
     \includegraphics{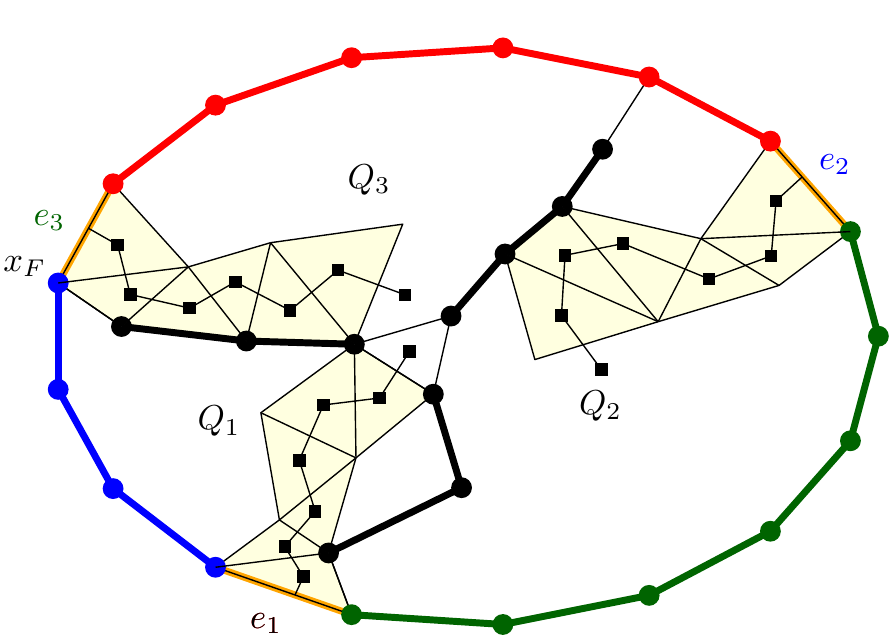}
   \end{center}
   \caption{Searching for the trichromatic triangle $\tau$ beginning at the portals $e_1$, $e_2$, and $e_3$. In this example, $\tau=t_{1,6}$ is found after $k+1=7$ steps.}
   \figlabel{fast-search}
 \end{figure}

Suppose the search for $\tau$ succeeds when $\tau=t_{i,k}$ in Step~$k$.  Thus, for each $i\in\{1,2,3\}$, the algorithm has searched the sequence of triangles $t_{i,0},\ldots,t_{i,k}$.  Each of the shorter subsequences $S_i:=t_{i,0},\ldots,t_{i,k-1}$ consists entirely of bichromatic triangles. Each sequence $S_i$ contains $k$ bichromatic triangles whose vertices are coloured with $\{c_1,c_2,c_3\}\setminus c_i$.

Refer to the second part of \figref{fast-search}. Consider again the graph $M$ with inner faces $Q_1$, $Q_2$, $Q_3$ and $\tau$. For each $i\in\{1,2,3\}$, each face in $S_i$ is contained in $Q_i$.  Since $f_i$ counts the number of triangular faces of $N$ contained in $Q_i$, this implies that $f_i\ge k$ for each $i\in \{1,2,3\}$.  Therefore, $\min\{f_1,f_2,f_3\}\ge k$.  On the other hand, the search for for $\tau$ took $1+k$ steps, each of which performs three $\textsc{NearestMarkedAncestor}(w)$ queries and therefore the entire search runs in time $O(1+k) \subseteq O(1+\min\{f_1,f_2,f_3\})$.

Next, the algorithm prepares the three subproblems defined by $Q_1$, $Q_2$, and $Q_3$ on which to recurse.  To do this it follows the path $Z_i$, in $T$, from each vertex $x_i$ of $\tau$ to the first vertex of $P_i$.  It colours each vertex of $z_i$ with some colour $c_4\not\in\{c_1,c_2,c_3\}$ and then walks $Z_i$ backward, calling $\textsc{Mark}(w)$ for each vertex of $Z_i$.

Finally, in preparing each subproblem $Q_i$ for the recursive invocation, it may be necessary to change the colour of an already coloured vertex $v$ of $F$ with the colour $c_4$ before making the recursive call and then recolouring $v$ with its original colour once the recursion is complete.  This corresponds to introducing an artificial portal adjacent to an edge of $\tau$ contained in $F$.

\subsection{Running-Time Analysis}

We analyze the running time of the preceding algorithm by analyzing two parts separately.

During each recursive invocation, the algorithm does work to find the trichromatic triangle $\tau$.  The time associated with this is $O(1+k)$ where $k\ge \min\{f_1,f_2,f_3\}$.  As already described above, this leads to a recurrence of the form $T(f) \le O(\min\{f_1,f_2,f_3\}) + T(f_1)+T(f_2)+T(f_3)$ which resolves to $O(f\log f)$.  In the initial call, $f=2n-3$ is the number of inner faces of $G$, so the total running time attributable to this part of the algorithm is $O(n\log n)$.

In addition to this, the algorithm does other work in preparing inputs for recursive calls.  Once $\tau$ is identified, the previously uncoloured vertices of $Y$ are coloured.  A vertex $v\in V(G)$ appears in $Y$ during exactly one recursive invocation. Thus, colouring the vertices of $Y$ contributes a total of $O(n)$ time to the running time of the entire algorithm.

Finally, the vertices of $Y$ are added to the set $S$ maintained by the nearest marked ancestor data structure using calls to $\textsc{Mark}(w)$.  By \lemref{nearest_marked_ancestor}, this takes a total of $O(n\log n)$ time.
This completes the proof of the following theorem:

\begin{thm}
  There exists an algorithm that, given any $n$-vertex triangulation $G$ and any breadth-first-search tree $T$ of $G$, runs in $O(n\log n)$ time and finds a partition $\mathcal{P}$ of $V(G)$ such that each $P\in\mathcal{P}$ induces a vertical path in $T$ and the quotient graph $H:=G/\mathcal{P}$ has treewidth at most $8$.
\end{thm}

\section{Discussion}
\seclabel{discussion}

Another variant of \thmref{triangulation-partition} described by Dujmović \etal\ gives a partition $\mathcal{P}$ of $V(G)$ such that $G/\mathcal{P}$ has treewidth at most 3 and each part $Y\in\mathcal{P}$ is the union of at most 3 vertical paths in $T$.  The algorithm described here also gives an $O(n\log n)$ time algorithm for this variant.

\subsection{Other Graph Classes}

\thmref{product-structure} has been generalized to a number of graph classes including bounded-genus graphs \cite{dujmovic.joret.ea:planar}, apex-minor free graphs \cite{dujmovic.joret.ea:planar}, graphs of bounded-degree from proper-minor closed families \cite{dujmovic.esperet.ea:planar}, and $k$-planar graphs \cite{dujmovic.morin.ea:structure}.  In all cases, these generalizations ultimately involve decomposing the input graph into a number of planar subgraphs and applying \thmref{product-structure} to each of these planar graphs.

In at least two cases, the extra work done in these generalizations can be done in $O(n\log n)$ time.  Combined with \thmref{main}, this gives $O(n\log n)$ time algorithms for the corresponding generalizations of \thmref{product-structure}.

\begin{itemize}
  \item For graphs $G$ of fixed Euler genus $g$, the result of Dujmović \etal\ \cite{dujmovic.joret.ea:planar} only requires finding a genus-$g$ embedding of $G$, computing a breadth-first-search tree $T$ of $G$, and computing any spanning-tree $D$ of the dual graph that does not cross edges of $T$.  The two spanning trees $T$ and $D$ can be computed in $O(n)$ time using standard algorithms.  The genus-$g$ embedding of $G$ can be computed in $O(n)$ time using an algorithm of Mohar \cite{mohar:linear}

  \item Given a $k$-plane embedding of a $k$-planar graph $G$, the result of Dujmović, Morin, and Wood \cite{dujmovic.morin.ea:structure} applies \thmref{product-structure} directly to the planar graph obtained by adding a dummy vertex at every point where a pair of edges crosses.

  While the problem of testing $k$-planarity of a graph is NP-complete, even for $k=1$ \cite{grigoriev.bodlaender:algorithms,korzhik.mohar:minimal,urschel.wellens:testing}, there are a number of graph classes that are $k$-planar and in which an embedding can be found easily.  These include (appropriate representations of) map graphs, bounded-degree string graphs, powers of bounded-degree planar graphs, and $k$-nearest-neighbour graphs of points in $\R^2$ \cite[Section~8]{dujmovic.morin.ea:structure}.
\end{itemize}

\subsection{Applications}

The algorithm presented here applies immediately to the four applications of \thmref{product-structure} discussed in the introduction.

\begin{itemize}
  \item There exists an algorithm that, given an $n$-vertex planar graph $G$, runs in $O(n\log n)$ time and computes a 49 queue layout of $G$ \cite{dujmovic.joret.ea:planar}.

  \item There exists an algorithm that, given an $n$-vertex planar graph $G$, runs in $O(n\log n)$ time and computes a nonrepetitive colouring of $G$ using at most 768 colours \cite{dujmovic.esperet.ea:planar}.

  \item There exists an algorithm that, given an $n$-vertex planar graph $G$, runs in $O(n\log n)$ time and computes a $(1+o(1))\log n$-bit adjacency labelling of $G$ \cite{dujmovic.esperet.ea:adjacency}.

  \item There exists an algorithm that, given an $n$-vertex planar graph $G$ and an integer $p$, runs in $O(p^3n\log n)$ time and computes $p$-centered colouring of $G$ using at most $3(p+1)\binom{p+3}{3}$ colours \cite{debski.felsner.ea:improved}.

\end{itemize}

Prior to this work, the bottleneck in all these algorithms was the $\Theta(n^2)$ worst-case running time of the algorithm for computing the decomposition of \thmref{product-structure}.

\subsection{Future Work}

The obvious open problem left by our work is that of finding a faster algorithm.  Can the running-time in \thmref{main} be improved to $O(n)$?

\section*{Acknowledgement}

Part of this research was conducted during the Eighth Workshop on Geometry and Graphs, held at the Bellairs Research Institute, January~31--February~7, 2020.  The author is grateful to the other organizers and participants for providing a stimulating research environment.

The author would like to thank Vida~Dujmović and our summer student \mbox{Ivana~Marusic}.  Ivana implemented the algorithm described in \secref{original} and the algorithm described in \secref{faster}.  The second implementation motivated the development of the simple nearest marked ancestor data structure in \appref{data_structures}.

\bibliographystyle{plainurl}
\bibliography{fasttri}

\appendix

\section{Data Structures}
\applabel{data_structures}

In this appendix we describe the simple data structures used by our algorithm. Devising these data structures is an exercise in the use of two standard techniques (a simple union-find data structure and the interval labelling scheme for rooted trees).

\subsection{Interval Splitting}

An interval splitting data structure stores an initially-empty subset $S$ of $\{1,\ldots,n\}$ under the following two operations, each of which takes an integer argument $x\in\{1,\ldots,n\}$:
\begin{itemize}
    \item $\textsc{Add}(x)$: Add $x$ to the set $S$, i.e., $S\gets S\cup \{x\}$.  It is a precondition of this operation that $x\not\in S$.
    \item $\textsc{Interval}(x)$: Return the pair $(i,j)$ where $i=\max\{y\in S\cup\{0\}:y<x\}$ and $j=\min\{y\in S\cup\{n+1\}:y\ge x\}$.
\end{itemize}

\begin{lem}\lemlabel{interval_splitting}
    There exists a data structure that preprocesses an integer $n$ and supports the $\textsc{Add}(x)$ and $\textsc{Interval}(x)$ operations.  The data structures uses $O(n)$ preprocessing time, each $\textsc{Interval}(x)$ operation runs in $O(1)$ time and any sequence of $\textsc{Add}(x)$ operations takes a total of $O(n\log n)$ time.
\end{lem}

\begin{proof}
    The data structure is essentially the inverse of one of the simplest union-find data structures that represents sets as linked lists in which each node has a pointer to the head of the list.

    The data structure contains an array $a_1,\ldots,a_n$ of pointers.  For each $x\in\{1,\ldots,n\}$, the array entry $a_x$ points to a memory location $m$ storing the interval $(i,j)$ that answers the $\textsc{Interval}(x)$ query.  In this way each $\textsc{Interval}(x)$ query operation runs in $O(1)$ time, as required.

    The data structure is memory-efficient in the following sense: Suppose that, at some point in time $S=\{x_1,\ldots,x_k\}$ with $x_1<\cdots<x_k$.   and use the convention that $x_0:=0$ and $x_{k+1}=n+1$.  Then, for each $i\in \{0,\ldots,k\}$, the array locations $a_{x_{i}+1},\ldots,a_{\min\{n,x_{i+1}\}}$ all point the same memory location $m$ that contains the pair $(i,j)$, $i<x\le j$ that answers the $\textsc{Interval}(x)$ query for each value $x\in\{x_{i}+1,\ldots,x_{i+1}\}$.  Thus, the number of distinct memory locations $m$ used to store answers to $\textsc{Interval}(x)$ queries is exactly $k+1$.

    To perform an $\textsc{Add}(x)$ operation, the data structure first looks at the pair $(i,j)$ stored at the memory location $m$ referenced by $a_x$.
    Observe that the $j-i$ array entries $a_{i+1},\ldots,a_{j}$ all point to the same memory location $m$ and let $k:=x-i$.  The data structure allocates a new memory location $m'$ containing a pair $(i',j')$.  The algorithm then makes a choice, depending on the value of $k$.

    \begin{enumerate}
        \item If $k\le (j-i)/2$, then it sets $(i',j')\gets (i,x)$, sets $(i,j)\gets (x,j)$ and sets $a_{i+1},\ldots,a_{x}\gets m'$.
        \item Otherwise, it sets $(i',j')\gets (x,j)$, sets $(i,j)\gets (i,x)$ and sets $a_{x+1},\ldots,a_{j}\gets m'$.
    \end{enumerate}
    It is straightforward to verify that these operations are correct.

    To analyze total running time of a sequence of $\textsc{Add}(x)$ operations, we use the potential method of amortized analysis.  For each $\ell\in\{1,\ldots,n\}$, let $\Phi_\ell=c\log_2(j-i)$ where $(i,j)$ is the answer to $\textsc{Interval}(\ell)$ and let $\Phi=\sum_{\ell=1}^n \Phi_\ell$.  Observe that $0\le \Phi_\ell\le c\log_2(n+1)$ for each $\ell\in\{1,\ldots,n\}$, so $0\le \Phi\le cn\log_2(n+1)$. Note, furthermore that the $\textsc{Interval}(x)$ operation has no effect on $\Phi$.

    When an $\textsc{Add}(x)$ operation runs, it updates some number $z$ of array entries (either $z=k$ or $z=j-i-k$.   This takes $O(z)$ time and does not cause $\Phi_\ell$ to increase for any $\ell\in\{1,\ldots,n\}$.  Furthermore, this operation causes $\Phi_\ell$ to decrease by at least $c$ for each array entry $a_\ell$ that is modified. Therefore, letting $\Phi$ and $\Phi'$ denote the value of $\Phi$ before and after this operation, we have $\Phi'-\Phi \le -cz$.  The amortized running time of this operation is therefore $O(1+z+\Phi'-\Phi) = O(1)$ for a sufficiently large constant $c$.

    Therefore each $\textsc{Add}(x)$ operation runs in $O(1)$ amortized time, the minimum potential is $0$ and the maximum potential is $cn\log_2(n+1)$, so the total running time of any sequence of $\textsc{Add}(x)$ operations is at most $O(m + n\log n)$.  The precondition that $x\not\in S$ ensures that $m\le n$, so the total running time is $O(n\log n)$.
\end{proof}

\subsection{Nearest Marked Ancestor}

\begin{proof}[Proof of \lemref{nearest_marked_ancestor}]
    The data structure is essentially the interval labelling scheme for trees combined with the interval splitting data structure from the previous section.

    Let $T'$ be the directed graph obtained by replacing each undirected edge $vw$ of $T$ with two directed edges $vw$ and $wv$.  Since every node of $T'$ has the same in and out degree, it is Eulerian.  Let $v_1,v_2,\ldots,v_{2n-1}$ be the sequence of vertices encountered during an Euler tour of $T'$ that begins and ends at the root $v_1=v_{2n-1}$ of $T$.
    For each node $v$ of $T$, let $i_v:=\min\{i:v_i=v\}$ and $j_v:=\max\{j:v_j=v\}$. Observe that $v_{i_v},v_{i_v+1},\ldots,v_{j_v}$ contains exactly those nodes of $T$ that have $v$ as a $T$-ancestor.

    The data structure stores the sequence $v_0,v_1,v_2,\ldots,v_{2n-1}$ in an array and maintains an interval splitting data structure on the set $1,\ldots,2n-1$.  The operation of $\textsc{Mark}(w)$ is simple: We simply call $\textsc{Add}(i_w)$ and $\textsc{Add}(j_w)$.  Note that the precondition that $w\not\in S$ ensures that the precondition $x\not\in S$ of the $\textsc{Add}(x)$ operation is satisified.  Thus, any sequence of $\textsc{Mark}(w)$ operations results in a sequence of $\textsc{Add}(x)$ operations on the set $\{1,\ldots,2n-1\}$.  By \lemref{interval_splitting} this takes a total of $O(n\log n)$ time, as required.

    The operation of the $\textsc{NearestMarkedAncestor}(v)$ is almost as simple.  We call $\textsc{Interval}(i_w)$ to obtain some pair $(i,j)$ with $i< i_w \le j$.  There are two cases to consider:
    \begin{enumerate}
        \item If $j=i_w$ then this is because $w\in S$, in which case $w$ is the nearest marked ancestor of itself.
        \item Otherwise, $i<i_w<j$, and $(i,j)=(i_v,j_v)$ for some node $v\in S$.  Therefore, $v_{i}=v_{j}$ is the nearest marked ancestor of $w$.
    \end{enumerate}
    Therefore, in either case, we obtain the nearest marked ancestor of $w$ in $O(1)$ time, as required.
\end{proof}

\end{document}